\documentclass[twocolumn]{article}

\usepackage{dblfloatfix}
\usepackage{epsfig}
\usepackage{graphicx}
\usepackage{amsfonts, amstext, amsmath, amssymb, mathrsfs}
\usepackage{mathtools}
\usepackage{amsthm}
\usepackage{float}
\usepackage{url}            

\newlength\titlebox 
\newtheorem{definition}{Definition}
\newtheorem{theorem}{Theorem}
\newtheorem{lemma}{Lemma}
\newtheorem{remark}{Remark}

\newtheorem{corollary}{Corollary}

\usepackage{blindtext}

\begin{document}
\title{{Active Community Detection with Maximal Expected Model Change}}
\author{ Dan Kushnir\thanks{Bell Laboratories, Nokia, Murray Hill, NJ 07974, USA. Email:dan.kushnir@nokia-bell-labs.com} \and Benjamin Mirabelli\thanks{Applied and Computational Mathematics
Princeton University
Princeton, NJ 08544, USA.
Email:benno.mirabelli@gmail.com}}
\date{}
\maketitle
\begin{abstract}
  We present a novel active learning algorithm for community detection on networks. Our proposed algorithm uses a Maximal Expected Model Change (MEMC) criterion for querying network nodes label assignments. MEMC detects nodes that maximally change the community assignment likelihood model following a query. Our method is inspired by detection in the benchmark Stochastic Block Model (SBM), where we provide  sample complexity analysis and empirical study with SBM and real network data for binary as well as for the multi-class settings. The analysis also covers the most challenging case of sparse degree and below-detection-threshold SBMs, where we observe a super-linear error reduction. MEMC is shown to be superior to the random selection baseline and other state-of-the-art active learners.
\end{abstract}

\section{Introduction}

Community detection, or clustering on networks, is a fundamental problem in a broad range of disciplines, from the study of biological and social networks to the classification of non-graphical data sets via the construction of pairwise similarity graphs \cite{comdetectBook}.  However, perfectly recovering each element's community by only observing the given graphical data is shown to be statistically impossible for many networks of interest.  Therefore, a growing area of research has been focused on the development of semi-supervised community detection algorithms \cite{Allahverdyan,Eaton,mossel_partial,Zhang}. More recently active learning \cite{Settles} has been introduced to the task of community detection (e.g. \cite{Cheng,Gadde,Moore}). Active learning, in this setting, allows to use a minimal and intelligently selected set of nodes to be labelled in order to improve community detection. Active learning is especially beneficial when labeling information and training are hard or expensive to obtain.

In this paper we focus on using active learning to solve both the problem of `detection' and 'recovery' of communities: where 'detection' is related to finding community assignments that have a non-trivial correlation with the true assignments, and recovery addresses the assignments to be correct with vanishing error probability. In the context of detection \cite{Decelle,Mossel_3} have shown that for benchmark random networks generated by the stochastic block model (SBM) \cite{SBM}, there exists a fundamental `detection threshold' for the unsupervised setting. These results prove that when the signal to noise ratio (SNR) of SBM-generated graphical data is below the `detection' threshold (SNR$<1$) then it is statistically impossible for any strictly graph-based algorithm to predict community assignments with better accuracy than random chance \cite{Mossel_3,abbe_rev}.
Recent work on clustering SBM networks with random label information \cite{saad,Xu,mossel_partial} (aka semi-supervised) has shown detection is possible below the threshold. We note \cite{mossel_partial} for the challenging case of sparse degree networks, where equivalence between label propagation in broadcast-trees and SBM is drawn, also covering the SNR$<1$ case.

Active learning in the context of community detection is a new and pioneering area, in particular in the context of SBMs. So far, only a handful of active learning strategies have been introduced to the field of community-detection. In the context of SBM we note entropy-based selection criterion of \cite{Moore}, which computes the difference between a node's uncertainty and its correlation with other graph nodes in order to select the best node to query. The computationally demanding Gibbs sampling is used to derive the conditional distribution of node labels in SBM. In \cite{Gadde} the authors address the perfect recovery problem in the SBM's logarithmic degree regime, which uses sampling of nodes whose neighbor's labels suggest the highest disagreement on their current assignment. The sampling is done once from the initial label approximation, and there is no update of the labeling after each query\textbackslash sample. Thus, hypothesis updates, which prevent redundant queries, are not exploited. Other network\textbackslash graph based active learning methods in other contexts than SBM involve criteria such as 'representative' sampling \cite{Leng}, uncertainty sampling \cite{Kushnir}, and a combination of thereof \cite{Cheng, Yang}. \cite{Cheng} targets edge queries, however, node queries typically provide more information.

In this paper we construct a novel active learner that embeds a semi-supervised community-detection algorithm based on optimizing the likelihood function. Specifically in the case of SBM, we find the maximum-likelihood (ML) community labeling given a SBM-generated graph and a subset of already known community labels. In this sense our work takes the SDP solution of \cite{Goemans} for the unsupervised ML problem one step further into the semi-supervised domain, by using the labeling information as constraints. Our active learner utilizes the ML statistical framework in order to query nodes that Maximize the Expected Model Change (MEMC) of the likelihood model given the graph, labels, and the hypothesized query information. The MEMC criterion is comprised of optimizing two components simultaneously: the overall change in the global likelihood model for each query outcome, and the likelihood of the query label given the current model. Model change has been used for active learning in other frameworks (e.g. \cite{EMC1,Settles,EMC2}) but is novel for ML and community detection.

MEMC's sample complexity analysis marks its advantage over the baseline Random selection criterion in both below and above SNR-threshold networks and for different average node degree regimes. In particular, the analysis proves MEMC's preference to correct erroneous assignments, and bounds its query search range, while the Random sampling range is unbounded. The theoretical analysis is validated with an empirical study of the sample complexity. In the below SNR regime, our analysis and simulations reveal a surprisingly fast super-linear error reduction phase comprised of multiple nodes label corrections following a query. The super-linear phase is followed by a linear reduction phase that lasts until full recovery and where MEMC has a clear advantage over Random. We conclude our work with numerical experiments with SBM networks which validated our analysis, as well as experiments with real social networks showing a clear advantage of MEMC over the baseline and other state-of-the-art active learning.


\section{SBM and the Likelihood Function}\label{Section:SBM}
We start with definitions related to the SBM and derive new results pertaining to its likelihood.
\begin{definition}\label{def:SBM}
The \textit{Stochastic Block Model} (SBM) with parameters $n$, $r$, $a$ and $b$ defines a random graph ensemble on $n$ nodes. Each node is assigned a label, uniformly at random, from a discrete set of $r$ labels.  If node $i$ and node $j$ are assigned the same label there exists an edge between them with independent probability $p=\frac{a}{n}$ and if node $i$ and node $j$ are assigned different labels there exists an edge between them with independent probability $q=\frac{b}{n}$.
\end{definition}
When the number of nodes in each community is equal we refer to the model as the \textit{Symmetric SBM} (SSBM).
\begin{definition}\label{def:M}
(Modified Adjacency Matrix) Given a graph $G$ on $n$ vertices and parameters $p$ and $q$, let $M=M(G,p,q)$ be the $n\times n$ `modified' adjacency matrix where $M_{ij}=\log{\frac{p}{q}}$ if $G$ has an edge between node $i$ and node $j$ and $M_{ij}=\log\Big(\frac{1-p}{1-q}\Big)$ if $G$ does not have an edge between node $i$ and node $j$.  Let $SBM(n,r,p,q)$ represent the probability distribution over matrices $M$.
\end{definition}
For completion we provide the fundamental definition of the detection Signal-to-Noise ratio for SBMs \cite{Decelle} and the seminal results on community detection of \cite{Massoulie,Mossel_1} for $r=2$:
\begin{definition}\label{def:SNR}
The Signal-to-Noise-Ratio (SNR) of $SBM(n,r,\frac{a}{n},\frac{b}{n})$ is defined as
\begin{equation}\label{eq:SNR}
  SNR\doteq \frac{(a-b)^2}{r(a+(r-1)b}.
\end{equation}
\end{definition}

\begin{theorem}\label{conj:detect_thresh}
(Detection Threshold) \cite{Massoulie,Mossel_1}: for r=2, if $(a-b)^2> 2(a+b)$ then one can a.a.s. find a bisection which is positively correlated with $SBM(n,2,\frac{a}{n},\frac{b}{n})$.
\end{theorem}
\begin{definition}
(Vector Labels). Let $\bigtriangleup_r$ be a set of $r$ unit-vectors representing the vertices of an $(r-1)$-simplex and let $\bigtriangleup_r^n$ be the set of all $n$-tuples of such vectors.  Thus, if ${v_i \brack v_j}\in \bigtriangleup_r^2$, then the inner-product
\[
\langle v_i,v_j\rangle=
\begin{cases}
-\frac{1}{r-1} &\text{if }v_i\neq v_j\\
\;\;\;\;\;1 &\text{if }v_i= v_j.
\end{cases}
\]
\end{definition}

From now on we associate any discrete labeling over $r$ communities with vectors in the set $\bigtriangleup_r$. Thus, we refer to a complete labeling over $n$ nodes as $X\in\bigtriangleup_r^n$, where $X_i\in \bigtriangleup_r$ is the $i^{th}$ node's label and corresponds to the $i^{th}$ row of $X$. In Theorem \ref{Theorem:Main} we derive the un-normalized probability of any discrete community labeling assignment $x\in\bigtriangleup_r^n$ as a monotonic and convex function of the quadratic-form $Tr(x^TMx)$. 

\begin{theorem}\label{Theorem:Main}
Let $X\in \bigtriangleup_r^n$ be a discrete labeling assignment and let $\textbf{M}\sim SBM(n,r,p,q)$.  Then
\begin{equation}\label{Probability:DiscreteUnsupservised}
\mathbb{P}[\textbf{X}=X|\textbf{M}=M] \propto e^{\frac{r-1}{2r}Tr(X^TMX)}.
\end{equation}
\end{theorem}
\begin{corollary}\label{Corollary:Supervised}
Let $X_U\in \bigtriangleup_r^{n-k}$ be a discrete community labeling assignment for a set of $n-k$ unsupervised nodes and let $X_L\in \bigtriangleup_r^{k}$ be the true labeling assignment for a subset of $k$ supervised nodes. Let $\textbf{M}\sim SBM(n,r,p,q)$ where the last $k$ rows correspond to the $k$ supervised nodes. Then,
\begin{equation*}\label{Probability:DiscreteSupservised}
\mathbb{P}[\textbf{X}_U=X_U|\textbf{M}=M,\textbf{X}_L=X_L] \propto e^{\frac{r-1}{2r}Tr({X_U\brack X_L}^TM{X_U\brack X_L})}.
\end{equation*}
\end{corollary}

\section{Semi-Supervised}\label{Section:SemiSupervised}

For the problem of semi-supervised community-detection, we seek the discrete labeling assignment $X_U\in \bigtriangleup_r^{n-k}$ that maximizes $\mathbb{P}[\textbf{X}_U=X_U|\textbf{M}=M,\textbf{X}_L=X_L]$ ($k$ being the the number of labeled nodes).  From Corollary \ref{Corollary:Supervised}, this is equivalent to finding the labeling assignment $X_U\in \bigtriangleup_r^{n-k}$ that maximizes $Tr({X_U\brack X_L}^TM{X_U\brack X_L})$.  However, in general, exactly finding this discrete ML labeling is exponentially hard and can naively take up to $O(r^{n-k})$-time. Thus, we relax the optimization problem into a SDP and find an `approximate' ML labeling in the relaxed domain.

While SDP has been used in the unsupervised case to solve community detection problems (see references in \cite{abbe_rev}), we present here a new formulation for the Semi-supervised case using the labeled set as a set of constraints in the original SDP.
Specifically, given the graph $\textbf{M}=M$ and labeling $\textbf{X}_L=X_L$, the SDP of our semi-supervised algorithm maximizes the convex function $Tr({X_U\brack X_L}^TM{X_U\brack X_L})$ over the `relaxed' domain $X_U\in \mathbb{R}^{(n-k)\times(n-k+r)}$ and $\lVert X_i\rVert_2=1$ for every $i$.  
We find the optimal assignment for $X_U$ in this relaxed domain by factoring the solution $\mathbb{X}={X_U\brack X_L}{X_U\brack X_L}^T$ of the following Semi-Definite Program:
\begin{equation}\label{Algo:SDP}
\begin{split}
\text{SDP}(M,X_L)&{: } \max_{\mathbb{X}} \text{ Tr}(M\mathbb{X})\;such\;that\\
&\mathbb{X} \succeq 0,\\
&\mathbb{X}_{ii} = 1 \text{ for } 1\leq i\leq n;\;\\&  \mathbb{X}_{ij} = 1 \text{ if } X_{L\_i}=X_{L\_j},\\
&\mathbb{X}_{ij} = -\frac{1}{r-1} \text{ if } X_{L\_i}\neq X_{L\_j};\;\;\;\;\;
\end{split}
\end{equation}
where $X_{U\_i}$ or $X_{L\_i}$ is just another way of writing $X_i$ while simultaneously specifying if node $i$ is currently in the unlabeled or labeled set.

We define the output of SDP$(M,X_L)$ to be the factorized matrix $X={X_U \brack X_L}$ rotated so that the vectors $X_L$ line up with their correct corresponding vectors in $\bigtriangleup_r$.  This SDP can be solved efficiently with programs such as Manopt \cite{Bandeira}. The factorization $\mathbb{X}={X_U\brack X_L}{X_U\brack X_L}^T$ can be found efficiently with Cholesky Decomposition and $X={X_U \brack X_L}$ has a unique rotated solution so long as there is at least one supervised label from each community (if not see Remark \ref{Remark:bestFitSimplex} below). From the constraints $\mathbb{X}_{ii} = 1$, each $X_i$ is a unit vector and referred to as the `vector-label' for node $i$.  To complete the semi-supervised algorithm we recover a discrete labeling by assigning each vector-label in $X_U$ to the closest corresponding vector-label in $\bigtriangleup_r$.

\begin{remark}\label{Remark:bestFitSimplex}
If $X_L$ contains less than $r-1$ distinct vectors, then rotating $X$ such that $X_L$ aligns with its corresponding vectors in $\bigtriangleup_r$ no longer produces a unique solution with respect to $\bigtriangleup_r$.  In this case we must use another algorithm to find the best-fit simplex for our data.  We present one such algorithm in the supplemental material.
\end{remark}

The semi-supervised algorithm (Fig. \ref{fig:semisup}) follows a relax-and-round procedure where the relaxed SDP (\ref{Algo:SDP}) is solved, followed by fitting a simplex to the SDP-output and rounding each SDP vector-label to its closest best-fit-simplex vector. Note that the function $unique(X) $ outputs the set of unique labels in $X$.
%

\begin{figure}[!htpb]
       \begin{tabular}{|rl|} \hline
       \vspace{.075cm}
       & \large \textbf{\textit{Semi-Supervised}}$(M,X_L,r)$\\
        &  \textbf{Input}: $M$: adjacency matrix, $X_L$: labeled set,  \\
        &\;\;\;\;\;\;\;\;\;\;\;\; $r$: number of communities\\
        & \textbf{Output}: $X\in \bigtriangleup_r^n$: complete labeling\\
        & 1. $X' =$ SDP($M,X_L$)\\
        & 2. If $|\text{unique}(X_L)|<r$\\
        &$\;\;\;\;\;\;\;\;$ $\bigtriangleup_{r} =$ bestFitSimplex$(X')$\\
        &$\;\;\;$ else\\
        &$\;\;\;\;\;\;\;\;$ $\bigtriangleup_{r} = $ unique$(X_L)$\\
        & 3. For $i=1$ to $n-k$\\
        &$\;\;\;\;\;\;\;\;$ $X_{U\_i} =\underset{X_j\in \bigtriangleup_r}{\text{argmax }} X'_{U\_i}X_j^T$\\
        & 4. $X = {X_U\brack X_L}$\\
        \hline
        \end{tabular}
\caption{The \textit{Semi-Supervised} algorithm.}\label{fig:semisup}
\end{figure}

\section{Active Learning}\label{Section:Active}
Our novel active learning strategy employs a querying criterion that is based on Maximal Expected Model Change (MEMC). In general, this strategy selects the unlabeled data point, $q\in U$, that if labeled is expected to cause the greatest change to some chosen model $\Phi$ with respect to some chosen norm $T$. 
Built on this idea, a tractable approximation of the distribution $\mathbb{P}[\textbf{X}_q=X_q|\textbf{M}=M,\textbf{X}_L=X_L]$ must be computed for all $q\in U$ and $X_q\in\bigtriangleup_r$, and a model $\Phi$ is to be chosen.

\subsection{Expectation}
In computing the \textit{expected} model change, the expectation must be taken with respect to each node's likelihood distribution over possible label assignments.  However, exactly computing this distribution, $\mathbb{P}[\textbf{X}_i=X_i|\textbf{M}=M,\textbf{X}_L=X_L]$, is an exponentially hard problem with complexity $O(r^{n-k})$ in general.  Hence, we use the classical strategy of ML-approximation \cite{Millar}, to approximate this distribution by
\begin{equation}\label{MLE}
\mathbb{P}[\textbf{X}_i=X_i|\textbf{M}=M,\textbf{X}_L=X_L,\textbf{X}_{U_{\neg i}}=X_{U_{ML}}],
\end{equation}
where $\textbf{X}_{U_{\neg i}}$ is the set of unknown node labels excluding node $i$, and $X_{U_{ML}}$ is the ML-labeling. For any labeling $\textbf{X}_{U_{\neg i}}=X_{U_{\neg i}}$ one can compute the above conditional probability with Lemma \ref{lemma:Conditional}:

\begin{lemma}
\label{lemma:Conditional}
Let $X_{U_{\neg i}}\in \bigtriangleup_r^{(n-k-1)}$ be a discrete community labeling assignment, let $\textbf{M}\sim SBM(n,r,p,q)$ and let $\textbf{M}_i$ be the $i^{th}$ row of $\textbf{M}$ without the $(i,i)$ entry.  Then, for any label $X_i\in\bigtriangleup_r$,
\begin{align}
\nonumber \mathbb{P}[\textbf{X}_i=X_i|\textbf{M}=M,\textbf{X}_L=X_L,\textbf{X}_{U_{\neg i}}=X_{U_{\neg i}}]\\ =\frac{e^{\frac{r-1}{r}(M_i{X_{U_{\neg i}}\brack X_L}X_i^T)}}{\sum_{X_j\in\bigtriangleup_r}e^{\frac{r-1}{r}(M_i{X_{U_{\neg i}}\brack X_L}X_j^T)}}.
\label{Equation:ConditionalDiscrete}
\end{align}
\end{lemma}

In order to approximate the ML-estimation defined in (\ref{MLE}), we calculate (\ref{Equation:ConditionalDiscrete}) for an approximate ML-labeling: 
we define a generalized ML-approximation of $\mathbb{P}[\textbf{X}_i=X_i|\textbf{M}=M,\textbf{X}_L=X_L]$ by generalizing (\ref{Equation:ConditionalDiscrete}) to  condition directly on the set of SDP-output vector-labels $X'_U\in\mathbb{R}^{(n-k)\times(n-k+r)}$. Thus, the new MLE-approximation for $\mathbb{P}[\textbf{X}_i=X_i|\textbf{M}=M,\textbf{X}_L=X_L]$ becomes
\begin{align}\label{Equation:SDP-MLE}
\nonumber \hat{\mathbb{P}}[\textbf{X}_i=X_i|\textbf{M}=M,\textbf{X}_L=X_L,\textbf{X}_{U_{\neg i}}=X'_{U_{\neg i}}]\\  = \frac{e^{\frac{r-1}{r}(M_i{X'_{U_{\neg i}}\brack X_L}X_i^T)}}{\sum_{X_j\in\bigtriangleup_r}e^{\frac{r-1}{r}(M_i{X'_{U_{\neg i}}\brack X_L}X_j^T)}},
\end{align}
where $X'_U$ is the output of SDP$(M,X_L)$.  Note that $\hat{\mathbb{P}}$ still defines a probability distribution but it is no longer conditioned on a discrete community labeling.

\subsection{Model Change}

 We use the ML-approximate distribution from (\ref{Equation:SDP-MLE}) to define the `model change' querying strategy.  First, we define the model to be the $(n\times r)$-matrix:\small
\begin{equation}\begin{split}
\Phi(M,X_L,\bigtriangleup_r)_{i,j}\doteq\\& \hspace{-19mm} \hat{\mathbb{P}}[\textbf{X}_i=X_{i}^{(j)}|\textbf{M}=M,\textbf{X}_L=X_L,\textbf{X}_{U_{\neg i}}=X'_{U_{\neg i}}]
\end{split}\end{equation}\normalsize
where, $X'_U=$ SDP$(M,X_L)$ and $X_i^{(j)}$ is the $j^{th}$ vector-label in $\bigtriangleup_r$.  For any already labeled node $i$, the probability distribution $\hat{\mathbb{P}}[\textbf{X}_{L\_i}=X_{L\_i}^{(j)}|\textbf{M}=M,\textbf{X}_L=X_L,\textbf{X}_{U_{\neg i}}=X'_{U_{\neg i}}]$ is defined to be the delta-function with mass $1$ on the node's true label.  
We define the model change, $\delta$, to be the sum of total variation distances between each node's probability distribution before and after a particular assignment $\textbf{X}_q=X_q$:
\begin{equation}
\delta(\Phi,X_q)\doteq \lVert\Phi(M,[X_L,X_q],\bigtriangleup_r)-\Phi(M,X_L,\bigtriangleup_r)\rVert_T.
\end{equation}

Thus, in order to maximize the Expected Model Change (EMC):
 \begin{equation}\label{eq:EMC}
\begin{split}
&\text{EMC}(M,X'_U,X_L,\bigtriangleup_r) \doteq \\
&\sum_{X_q\in \bigtriangleup_r} \hat{\mathbb{P}}[\textbf{X}_q=X_{q}|\textbf{M}=M,\textbf{X}_L=X_L,\textbf{X}_{U_{\neg q}}=X'_{U_{\neg q}}]\\ &\hspace{60mm} \cdot \delta(\Phi,X_q)
\end{split}
\end{equation}\normalsize
 we query the label of node
 \begin{equation}\label{Equation:MEMC}\begin{split}
  \hspace{-10mm}q=\text{MEMC}(M,X'_U,X_L,\bigtriangleup_r)\\& \hspace{-30mm}= \underset{q\in U}{\text{argmax }}EMC(M,X'_U,X_L,\bigtriangleup_r),
 \end{split}
\end{equation}\normalsize

\subsection{The algorithm}
Pseudo-code for the entire active MEMC algorithm is presented in Fig. \ref{fig:algs}.
\begin{figure}[!htpb]
\center
         \begin{tabular}{|rl|}\hline
       & \large \textbf{\textit{MEMC-Active}}$(M,X_L,r,Q)$\\
        &  \textbf{Input}: $M$, $X_L$, $r$, $Q$: query budget\\
        & \textbf{Output}: $X\in \bigtriangleup_r^n$\\
        & 1. For $queried =1$ to $Q$\\
        &$\;\;\;\;\;\;\;\;$a) $X' =$ SDP$(M,X_L)$\\
        &$\;\;\;\;\;\;\;\;$b) If $|\text{unique}(X_L)|<r$\\
        &$\;\;\;\;\;\;\;\;\;\;\;\;\;\;\;\;$i) $\bigtriangleup_{r} =$ bestFitSimplex$(X')$\\
        &$\;\;\;\;\;\;\;\;\;\;\;\;\;\;\;\;$ii) $q =\text{Anchor}(M,X'_U,X_L,\bigtriangleup_r)$\\
        &$\;\;\;\;\;\;\;\;\;\;\;$ Else\\
        &$\;\;\;\;\;\;\;\;\;\;\;\;\;\;\;\;$i) $\bigtriangleup_{r} = $ unique$(X_L)$\\
        &$\;\;\;\;\;\;\;\;\;\;\;\;\;\;\;\;$ii) $q =\text{MEMC}(M,X'_U,X_L,\bigtriangleup_r)$\\
        &$\;\;\;\;\;\;\;\;$c) $X_q=\text{Label}(\textbf{X}_q)$\\
        &$\;\;\;\;\;\;\;\;$d) $X_L={X_L \brack X_q}$\\
        & 2. $X = \text{Semi-Supervised}(M,X_L,r)$\\
        \hline
       \end{tabular}
\caption{The \textit{Active Learning} MEMC algorithm}
\label{fig:algs}
\end{figure}
The procedure $\text{Anchor}(M,X'_U,X_L,\bigtriangleup_r)$ is a querying mechanisms for the early stages when $X_L$ may not yet contain all existing community labels.  We further elaborate on the anchor nodes selection and additional speedup steps for the SDP in the supplementary material.

\section{Sample Complexity Analysis}
\label{sec:optimality}
In this section we provide theoretical sample complexity analysis to the MEMC criterion. As will be shown below, and in our empirical validation in Fig. \ref{Figure:SNR_SMALL1}, the error reduction rates that MEMC exhibits for the SBM differ throughout the active learning process and are characterized by 2 to 3 different phases, depending on the error type and the SNR regime. Therefore, our sample complexity analysis provides a more informative picture by proving the number of queries needed to recover the communities by MEMC, instead of the arbitrary error complexity analysis that is typical for active learning in standard classification tasks.

To facilitate our analysis we consider here the 2-community symmetric - $SSBM(n,2,a,b)$ where $p=\frac{a}{n}$ and $q=\frac{b}{n}$. Our results do not loose their utility for $r>2$ communities.  We consider in the analysis the two different settings of SBM graphs: above and below the SNR detection threshold (see definition in Theorem \ref{conj:detect_thresh}).

To set the stage for our analysis we first define the differential degree of a node with respect to an assignment $\tilde{X}$ of labels:
\begin{definition}
The \textit{Differential degree} of a node $v_i$ w.r.t to labeling $\tilde{X}$ is defined as
\begin{equation}
\begin{split}
d_{\tilde{X}}(v_i) = |\{v_j \in V| v_j\sim v_i,\tilde{X}_i = \tilde{X}_j\}| -\\
 |\{v_j \in V| v_j\sim v_i,\tilde{X}_i \neq \tilde{X}_j\}|
\end{split}
\end{equation}
\end{definition}
A node $v$ has a \textit{majority} with respect to a labeling $\tilde{X}$ if $d_{\tilde{X}}(v)>0$. Otherwise, $v$ has a \textit{minority}. We characterize two types of possible errors in the labeling approximation $\tilde{X}$ with respect to the true labels $X$:
\begin{definition}\label{def:typ1}
Type-1 error: $\tilde{X}_i\neq X_i$ where $v_i$ has a majority with respect to $X$.
\end{definition}
\begin{definition}\label{def:typ2}
Type-2 error: $\tilde{X}_i\neq X_i$ where $v_i$ has a minority with respect to $X$.
\end{definition}
Clearly, type 1 and type 2 error nodes comprise the total of possible errors. As will be shown below, the sample complexity of each error type is qualitatively different and depends on the stage of the active learning process, or, more specifically, on the distribution of the currently unlabeled nodes differential degree.

We note that in the problem setting of SBM the Maximum-a-Posteriori (MAP) estimator is such that it assigns the label
\begin{equation}
\tilde{X}_i = \arg \max_{\tilde{X}_i} P\{X_i =\tilde{X}_i|S=s,R=s \},
\end{equation}
where $S$ and $R$ are the random variables representing the number of neighbors of $v_i$ that have similar and opposite labels, respectively. Because label assignments are an independent process in SBM, the MAP estimator is equal to the ML estimator:
\begin{equation}
\tilde{X}_i = \arg \max_{\tilde{X}_i} P\{S=s,R=s|X_i =\tilde{X}_i \}.
\end{equation}
  Therefore, in order to facilitate our analysis we will consider the local MAP classier (shown above to be equivalent to the ML classier), and scalar labels.
\begin{figure}
\begin{center}
\begin{tabular}{c c}
\includegraphics[height=1.4in, width =1.5in]{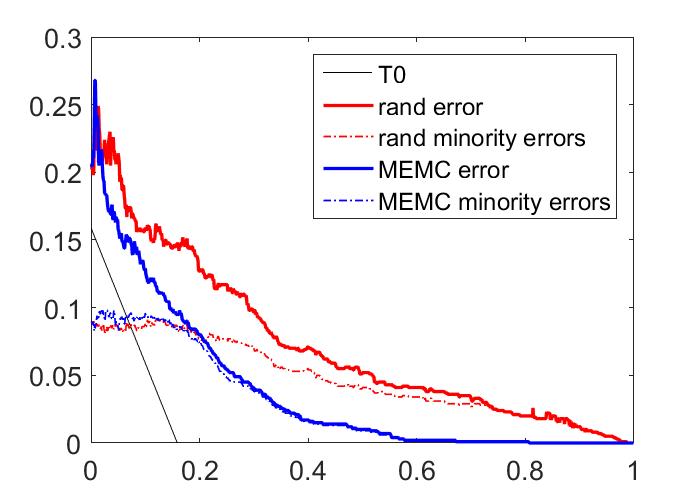}&\hspace{-7mm}
\includegraphics[height=1.4in, width =1.5in]{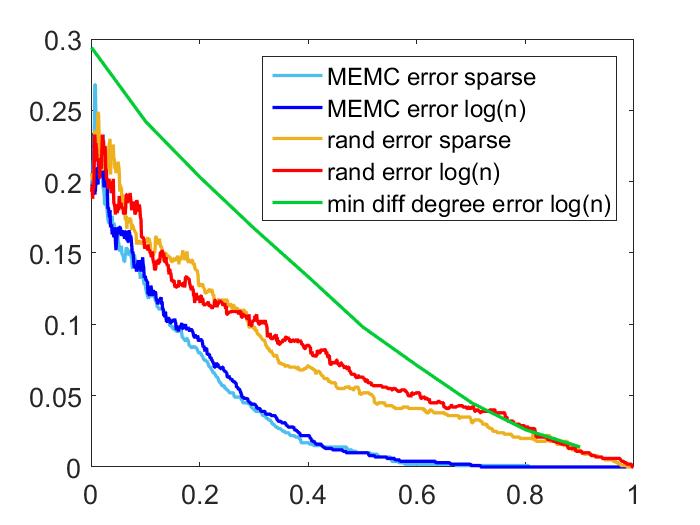}\\
(A)&\hspace{-7mm}(B)\\
\includegraphics[height=1.4in, width =1.5in]{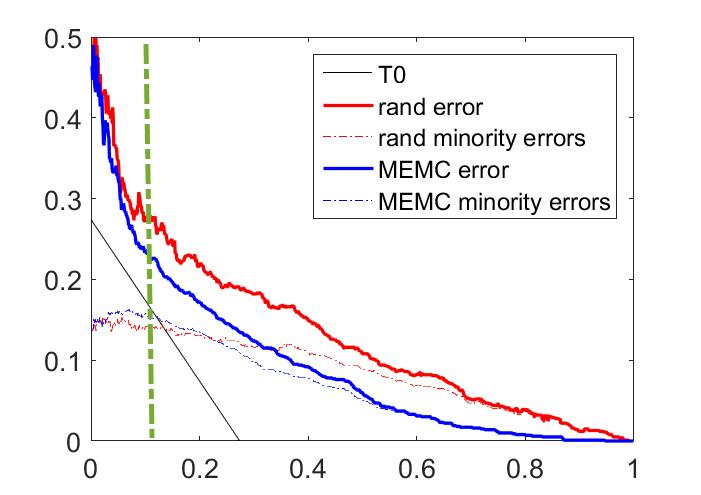}&\hspace{-7mm}
\includegraphics[height=1.4in, width =1.5in]{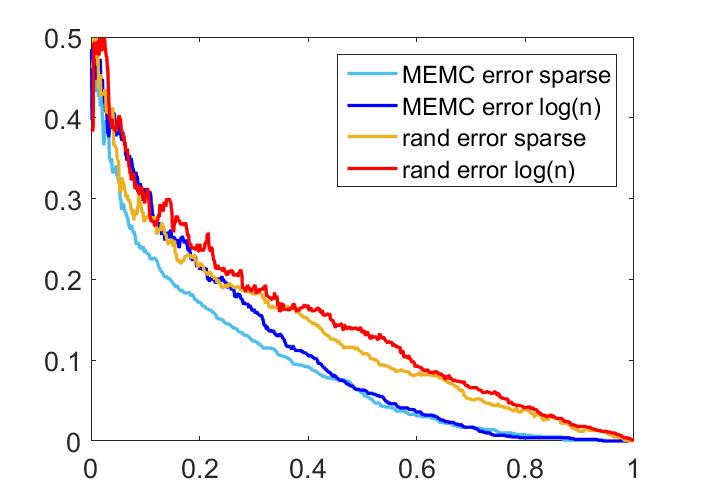}\\
(C)&\hspace{-7mm}(D)\\
\end{tabular}
\end{center}
\caption{Comparison of MEMC and random error vs. fraction of data queried. \textbf{A,B:} $SNR\approx 1.3$, sparse (A) logarithmic (B). \textbf{C,D:} $SNR\approx 0.6$ sparse (C) logarithmic (D). Dashed vertical line indicating transition from super-linear to linear phase.}
\label{Figure:SNR_SMALL1}
\end{figure}
\subsection{The SNR $>$ 1 case}
For the case of SNR$>1$ we provide our main result in Theorem \ref{thm:sampComplxHSNR} for recovery.
Theorem \ref{thm:sampComplxHSNR} is also validated by our empirical observations in Fig. \ref{Figure:SNR_SMALL1}-A,B in section \ref{Section:Results}. In particular, our experiments in Fig. \ref{Figure:SNR_SMALL1}-A,B validate that there are 2 different phases in which different types of errors are corrected at different rates throughout the recovery via the active learning process. This observation encourages us to examine the sample complexity of a full recovery to provide a comprehensive analysis.
\begin{theorem}
\label{thm:sampComplxHSNR}
(Sample Complexity) Consider the network $M\sim SSBM(n,2,p,q)$ such that $SNR>1$. Let $\tilde{X}=SDP(M,X_L)$ be the approximate scalar labels of the nodes, and $X_L$ is a set of known labels for a subset of nodes. Let the total of number of errors in $\tilde{X}$ be $m=m_1+m_2$, where $m_i$ corresponds to type-$i$ error. Then the sample complexity of MEMC for full recovery is at most
\begin{equation}
\label{eq:sampComplxHSNR_MEMC}
m_1 + (n-m_1)\Big(\exp \Big(-\Big(\sqrt{\frac{b}{2}}-\sqrt{\frac{a}{2}}\Big)^2\Big)+\sum_{k=1}^{-l_c}P(k;a,b)\Big),
\end{equation}
where $l_c$ is the critical differential degree such that a.a.s no nodes exists s.t. $d_X(v)<l_c$, and $P(k;a,b)$  is the Skellam distribution representing the probability of a difference between two Poisson random variables with means $a$ and $b$. Specifically, $l_c = inf\{k|(P(k;a,b)\leq k)=o(n^{-1})\}$. On the other hand Random will require order $n$ queries.
\end{theorem}
\noindent \textbf{Proof sketch:} The proof relies on three important auxiliary results - Lemmas \ref{thm:MEMC}, \ref{lemma:skellam}, and corollary \ref{corollary:SNRb1_typ2_vs_correct}  (the full proof of Theorem \ref{thm:sampComplxHSNR} and the auxiliary results are provided in the supplementary material). First, in Lemma \ref{thm:MEMC} we prove that type-1 errors have higher model change than other nodes of similar degree. Therefore, MEMC will initially query the type-1 error nodes, corresponding to the first $m_1$ component in (\ref{eq:sampComplxHSNR_MEMC}). Second, following the correction of type-1 error nodes MEMC triggers a search around zero differential degree for type-2 nodes which show preference for correcting type-2 errors over querying correctly assigned nodes, as per Corollary \ref{corollary:SNRb1_typ2_vs_correct}. Lastly, the type-2 error search range is proved to be symmetrically bounded around zero for a differential degree range that can be directly inferred from $a$ and $b$'s values as provided in Lemma \ref{lemma:skellam}. The random criterion will select all types of nodes at equal probability within an unbounded range of the possible differential degrees.

The sample complexity of the different algorithms for SNR$>1$ is summarized in Fig. \ref{fig:SNRG1prob}, where we also plot the curve for a hypothetic optimal active learner which makes only T0 errors of type-2 and queries only them. We provide below the auxiliary results for prooving Theorem \ref{thm:sampComplxHSNR}.

\begin{lemma}
\label{thm:MEMC}
Consider the network $M\sim SSBM(n,2,p,q)$. Let $\tilde{X}=SDP(M,X_L)$ be the scalar labels of the nodes, and $X_L$ is a set of known labels for a subset of nodes. Let $v^1$ be a node of differential degree $\delta>0$ for which $\tilde{X}$ has made a type-1 error, let $v^2$ be a node of differential degree $-\delta$ for which $\tilde{X}$ has made a type-2 error, and let $v^3$ be a correctly assigned node with differential degree $\delta$. Then $EMC$ satisfies
\begin{equation}
EMC(v^1)> EMC(v^2)= EMC(v^3).
\end{equation}
\end{lemma}

Next, we prove the conditions under which MEMC will have preference to correct type-2 nodes over querying correctly assigned nodes. In particular, Corollary \ref{corollary:SNRb1_typ2_vs_correct} guarantees that if there are type-2 error nodes with absolute differential degree that is lower than other node's differential degree (as typical for type-2 nodes), they will be queried by MEMC:
\begin{corollary}
\label{corollary:SNRb1_typ2_vs_correct}
Given $M\sim SSBM(n,2,a,b)$. Let $v^2_j$ be a type-2 node, and let $v^3_i$ be a correctly assigned node. If $|d_{\tilde{X}}(v^2_j)|<|d_{\tilde{X}}(v^3_i)|$ then $EMC(v^2_j)>EMC(v^3_i)$.
\end{corollary}

Lastly, Lemma \ref{lemma:skellam} provides bounds on the probability of having minority nodes which cause type-2 errors. While majority nodes can be classified correctly by MAP and ML classifiers once most of their neighbors are correctly identified, minority error nodes (i.e. type-2 error nodes) need to be queried directly in order to be corrected. Therefore estimating their proportion in the overall set is crucial for bounding sample complexity analysis for querying type-2 error nodes. In particular, Lemma 3 provides a bound on the probability of a minority node that can be used to derive a bound on the minimal differential degree $l_c$ that minority nodes could have. 
\begin{lemma}
\label{lemma:skellam} Given $SSBM(n,2,a,b)$, the probability of a minority node $P(c_{out}\geq c_{in})$ satisfies
\begin{equation}
\begin{split}
\label{eq:skellam}
&\frac{\exp(-(\sqrt{\frac{b}{2}} - \sqrt{\frac{a}{2}})^2)}{(\frac{a+b}{2})^2} - \frac{\exp(-(\frac{b}{2} +\frac{a}{2}))}{\sqrt{2}ab} - \frac{\exp(-(\frac{b}{2} +\frac{a}{2}))}{2ab}\\& \leq P(c_{out}\geq c_{in}) \leq \exp\Big(-\Big(\sqrt{\frac{b}{2}} - \sqrt{\frac{a}{2}}\Big)^2\Big).
\end{split}
\end{equation}
\end{lemma}

%

\begin{figure}
\hspace{-6mm}        \includegraphics[scale=0.25]{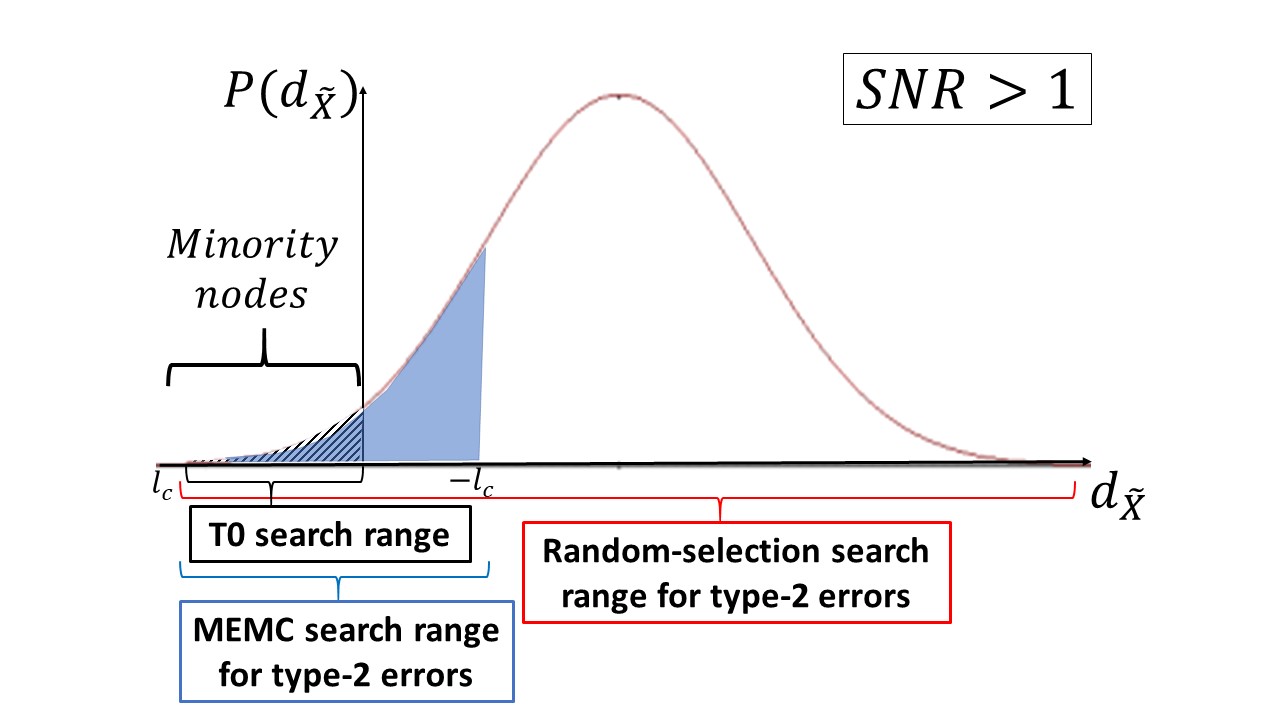}
 \label{fig:prob_sep}
\caption{Sample complexity for recovery in the $SNR>1$ regime. The probability mass of the differential degree $d_{\tilde{X}}$ is plotted, with regions marked for minority nodes and MEMC search region within critical degrees $[l_c,-l_c]$.}
\label{fig:SNRG1prob}
 \end{figure}

\subsection{The SNR $<$ 1 case}
Below the SNR threshold there is no unsupervised consistent estimator that can provide clustering better than a trivial guess \cite{Mossel_3}. In the semi-supervised setting, we observe a super-linear reduction in error with each query at the initial stage of querying (see Fig. \ref{Figure:SNR_SMALL1} C and D). This fast improvement rate is due to the concentration of the differential degree (for $\tilde{X}$) at zero, causing every query-flip to trigger cascades of label flips. We first define the phenomena of \textit{cascades}, and quantify its rate in Lemma \ref{lemma:cascades} below.
\begin{definition}\label{def:cascades} (Cascades)
  Consider the network $M\sim SSBM(n,2,p,q)$. Let $\tilde{X}=Alg_\Psi(M,X_L)$ be the predicted labels of the nodes according to a recovery algorithm $\Psi$, and $X_L$ is a set of known labels for a subset of nodes. Let $v_q$ be a node whose label has been queried. Then a cascade occurs if there exists at least one path $v_{q_0}\sim v_{q_1},...,v_{q_{l-1}}\sim v_{q_l},\;where\; v_{q_0}=v_q,\;v_{q_i}\sim v_{q_{i+1}},\;and \;l\geq 2$, such that $\tilde{x}_{q_i}\neq \tilde{x}^{up}_{q_i}$ for every $0\leq i\leq l$, where $\tilde{X}^{up}  = Alg_\Psi(M,X_L \cup x_q)$.
\end{definition}

Cascades give rise to the super-linear error reduction observed in the $SNR<1$ regime, where the expected differential degree with respect to $\tilde{X}$ is concentrated at and closely around 0. Lemma \ref{lemma:cascades} provides the estimation of the cascades effect on the overall community detection accuracy as a function of $a$ and $b$:
\begin{lemma}
\label{lemma:cascades}
Consider the connected network  $M\sim SSBM(n,2,p,q)$ with average degree $d$, and $SNR<1$, such that in expectation the differential degree $d_{\tilde{X}}(v)$ with respect to the $\tilde{X}$ approximation is zero. Then the expected number of majority(minority) nodes that correctly change their assignment after a query $v_q$ followed by a MAP update is at least
\begin{equation}\label{eq:cascades_maj}
  N_{maj(min)} = \frac{dp_{maj(min)}}{1-dp_{maj(min)}}(1-(dp_{maj(min)})^{\log_d(n)}),
\end{equation}
where $p_{maj(min)}$ is the probability of a majority(minority) node flipping its label correctly following its neighbor change of label and a MAP update. $p_{maj} = (1-2\bar{p}_{min})\cdot\Big(\frac{a-b}{4(a+b)}\Big)$, where $\bar{p}_{min}$ is assigned with the upper bound for the probability of a minority node in Equation (\ref{eq:skellam}), and, $p_{min} = \bar{p}_{min}\cdot\Big(\frac{b-a}{4(a+b)}\Big)$.
\end{lemma}

We note the dependency of (\ref{eq:cascades_maj}) in $n$ in the sparse regime: if $dp_{maj(min)}\approx d$, then$ (1-(dp_{maj(min)})^{\log_d(n)}\approx O(n)$ which will result in an order $n$ correction! 

Following the super linear error reduction phase, MEMC exhibits a linear reduction with preference to type-1 nodes succeeded by a bounded search for type-2 nodes, similarly to what was shown for the SNR$>1$ case. We note that these 3 stages of active learning are also observed in Fig. \ref{Figure:SNR_SMALL1}-C,D. We provide our main results on the sample complexity for recovery in the SNR $<$ 1 case in the following Theorem:
\begin{theorem}
\label{thm:sampComplxLSNR}
(Sample Complexity) Given the network  $M\sim SSBM(n,2,p,q)$ with average degree $d$, and $SNR<1$. Let $m=m_1+m_2$ be the total number of errors in $\tilde{X}$, where $m_i$ corresponds to errors of type-i. Then the expected sample complexity of MEMC is
\begin{equation}\begin{split}\label{eq:sampComplxLSNR_MEMC}
&\frac{n}{2(N_{maj}+N_{min})}+\Big(m_1 - \frac{nN_{maj}}{2(N_{maj}+N_{min})} \Big)\\& + \Big(n-\frac{nN_{min}}{2(N_{maj}+N_{min})} - m_1\Big)\\& \cdot\Big(\sum_{k=1}^{-l_c}P(k;a,b)+\exp\Big(-\Big(\sqrt{\frac{b}{2}}-\sqrt{\frac{a}{2}}\Big)^2\Big)\Big),
\end{split}\end{equation}
where $N_{maj(min)}$ are as defined in Lemma \ref{lemma:cascades}. Conversely, the sample complexity of Random is
$n$.
\end{theorem}

Theorem \ref{thm:sampComplxLSNR} concludes that MEMC number of queries is significantly lower than the Random criterion, even though both MEMC and Random demonstrate initially a super-linear error reduction.
To this end, there are two important observations that address MEMC's superiorly:  First, following Theorem \ref{thm:sampComplxLSNR}, type-1 nodes are recovered initially and the search for type-2 nodes is within a bounded differential degree range $[-l_c,l_c]$, where they are concentrated. Thus, MEMC search covers only a fraction of the $n$ nodes, while it still guarantees recovery of the type-2 errors. On the other hand, Random slection samples the full range of differential degrees for both types of error nodes. Second, in the sparse regime multiple connected components may exist. MEMC has preference to select the components of largest diameter, which will introduce higher model change and therefore deeper correction cascades, as seen in Fig. \ref{Figure:SNR_SMALL1} C,D. On the other hand, Random does not inflict preference for larger connected components.


\begin{figure*}
\begin{center}
\begin{tabular}{c c c c c}
\hspace{-5mm}\includegraphics[height=1.2in, width =1.3in]{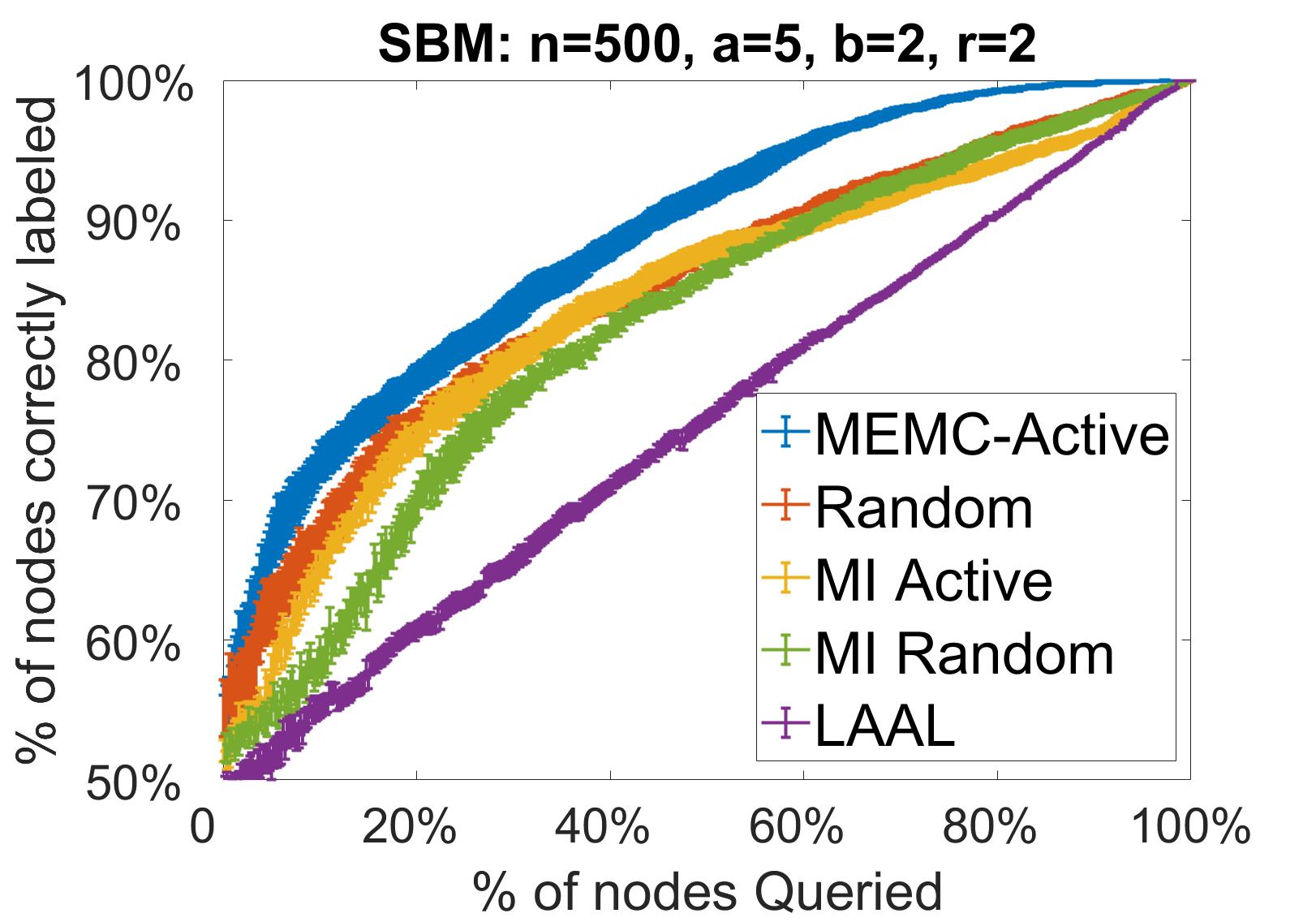}&
\includegraphics[height=1.2in, width =1.2in]{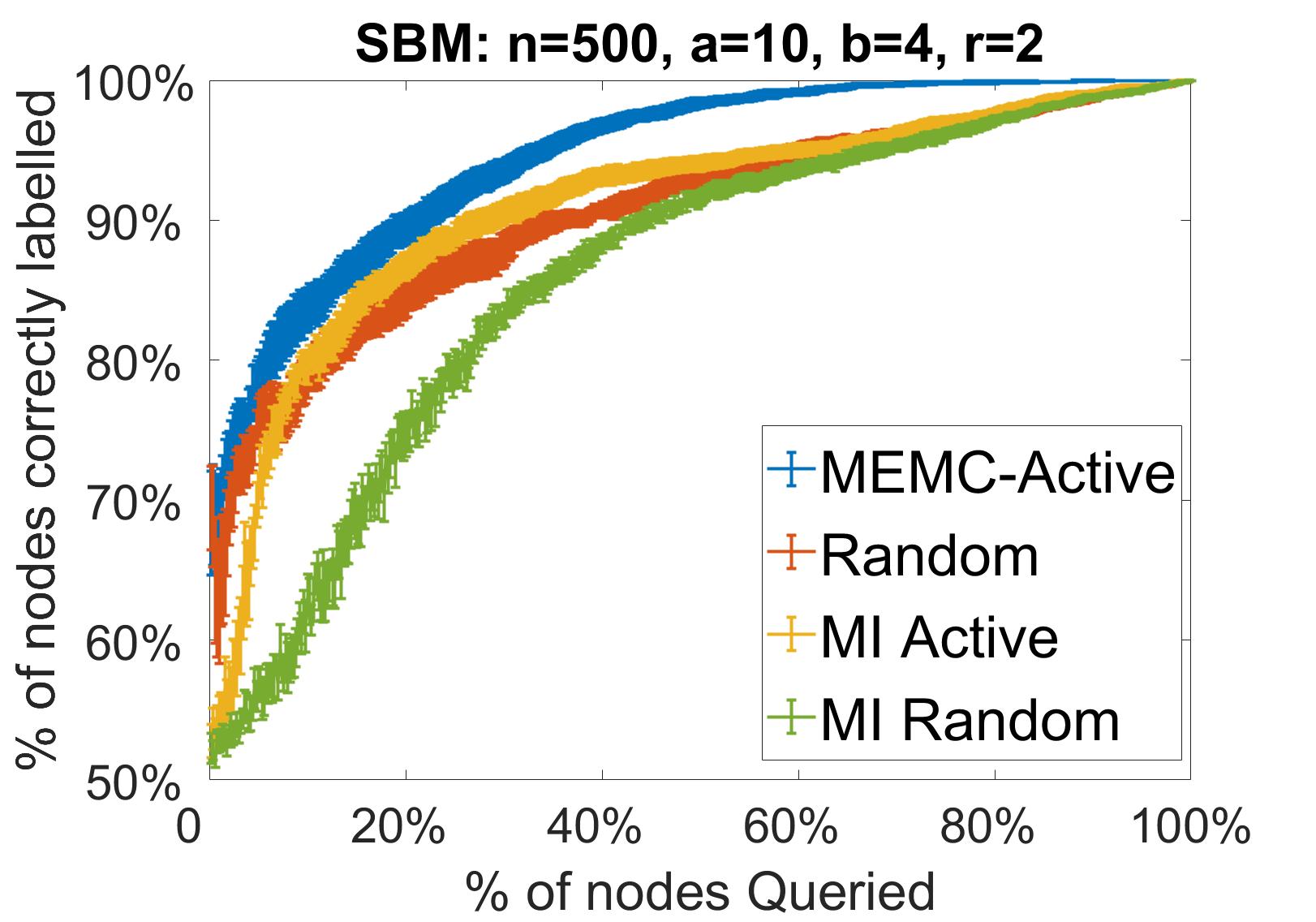}&
\includegraphics[height=1.2in, width =1.2in]{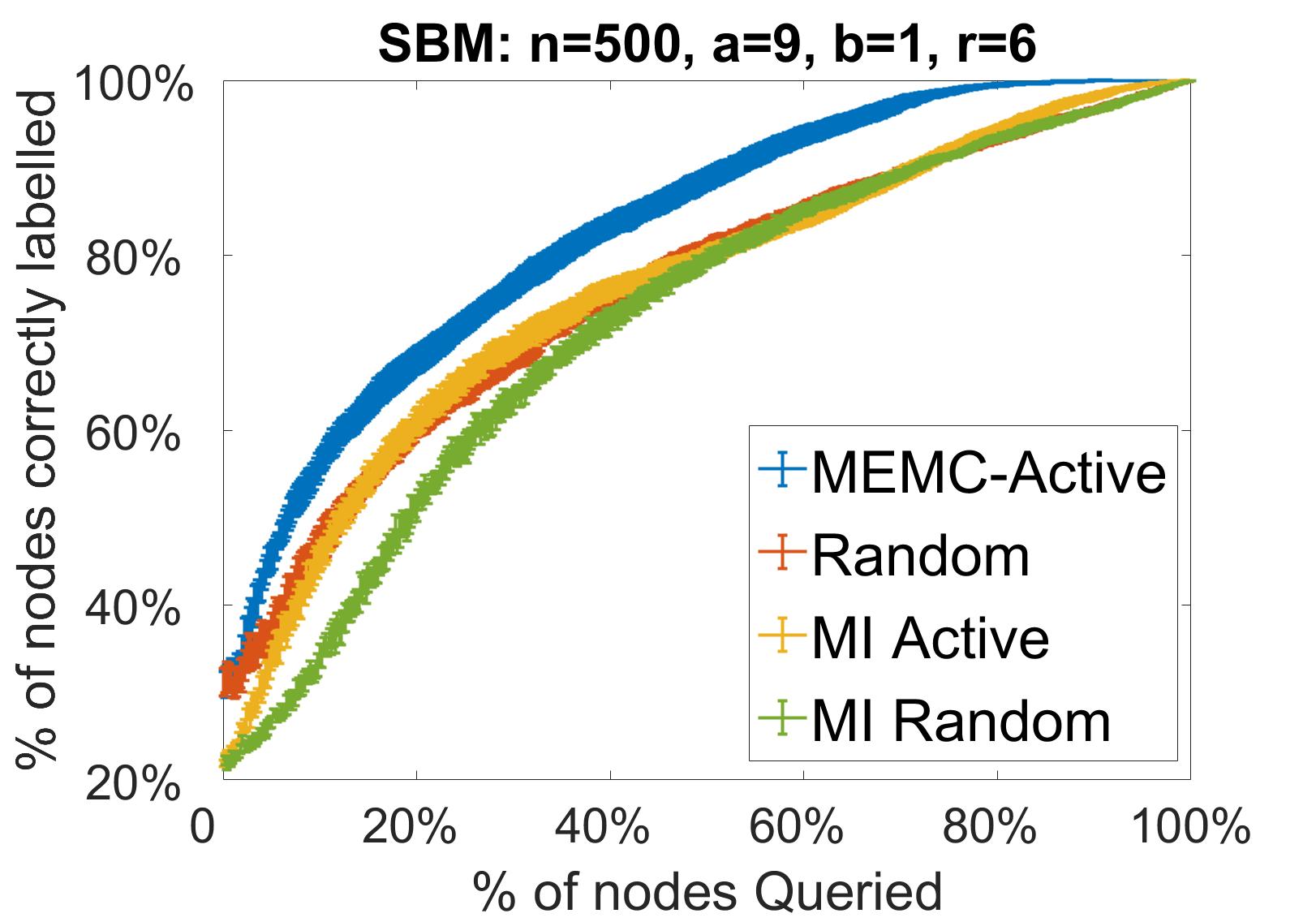}&
\includegraphics[height=1.2in, width =1.2in]{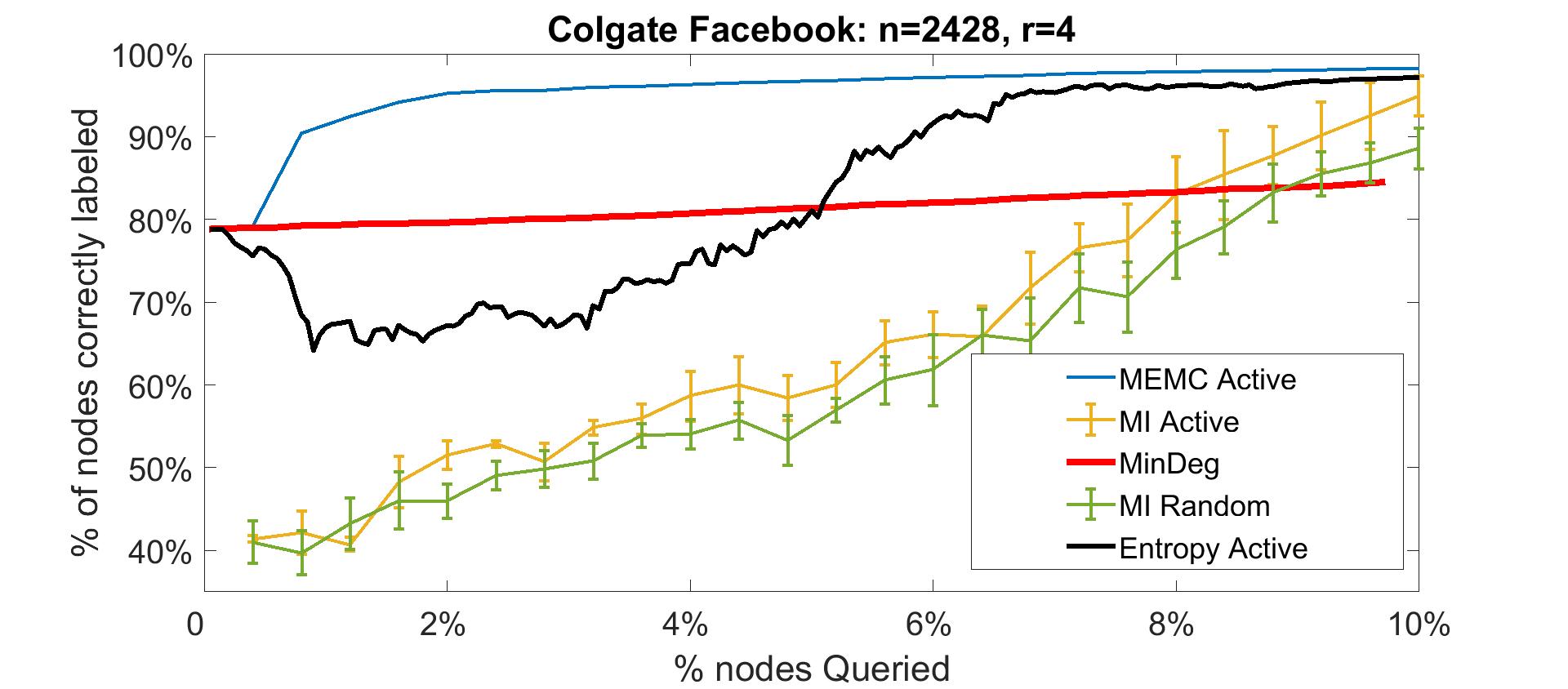}&
\includegraphics[height=1.2in, width =1.2in]{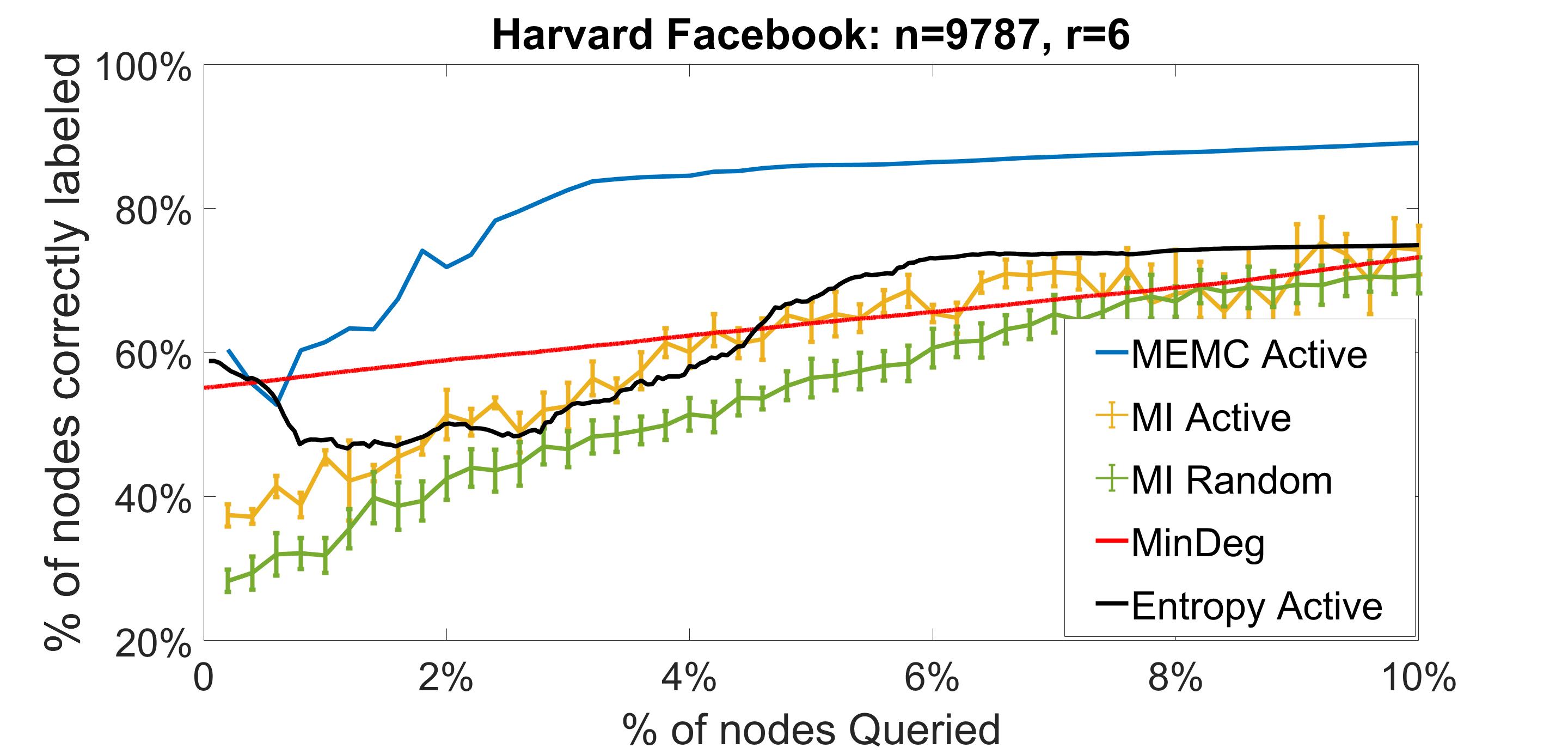}\\
(A) & (B)& (C) & (D)& (E)\\
\end{tabular}
\end{center}
\caption{Accuracy vs. Percentage of nodes queried. SBM: (A) $a=5$, $b=2$, $r=2$, (B) $a=10$ $b=4$ $r=2$, (C) $a=9$, $b=1$, $r=6$. Social (D) Colgate FB, $r=4$, (E) Harvard FB, $r=6$.}
\label{Figure:Plots}
\end{figure*}
\subsection{Logarithmic vs. Sparse Degree Regimes}
In the $SNR>1$ regime both sparse and logarithmic degree regimes have similar behavior (see Fig. \ref{Figure:SNR_SMALL1} A and B).  This similarity is a direct result of the general preference of MEMC to correct erroneous nodes. Moreover, since cascades do not take place in the $SNR>1$ regime, the degree regime, which influences the diameter of connected components, plays no role here.

In the $SNR<1$ case in the sparse regime the probability of zero differential degree is significantly higher than that of the logarithmic regime where the differential degree distribution is much more spread-out due to the higher degree magnitude. In the sparse degree regime significantly more mass of the distribution is localized at zero because $a$ and $b$ are small constants. As a result, correction cascades will occur with higher probability in the sparse regime. 


\section{Experimental Results}\label{Section:Results}
\noindent \textbf{Empirical sample-complexity study.}
To validate our theoretical analysis we provide an empirical study of the optimality of MEMC. The study comprises of an error analysis including type-1,2 error evolution. MEMC's curves are also compared with those of a Random selection.  We also provide the curve for a hypothetic optimal active algorithm which makes only T0 errors of type-2 and queries only them. The study covers SNR that is above (Fig. \ref{Figure:SNR_SMALL1} A, B) and below  (Fig. \ref{Figure:SNR_SMALL1} C, D) the detection threshold, and the cases of sparse vs. logarithmic average degree regimes. We underline in the following the main observations:
\begin{enumerate}
\item $SNR>1$: Type-2 errors remain fixed as type-1 errors are corrected, confirming Lemma \ref{thm:MEMC}. Random selection, conversely, presents a slower error improvement, due to it's unbounded search. Note T0 is larger than actual type-2 errors since some minority nodes are correctly classified by chance.
\item $SNR>1$: MEMC error curve aligns early with the type-2 error curve indicating transition to querying type-2 nodes, confirming Corollary \ref{corollary:SNRb1_typ2_vs_correct}.
\item $SNR>1$: The error curve for the minimal differential degree sampling methodology of \cite{Gadde} (designed for SNR$>1$ at logarithmic regime) is presented in Fig. \ref{Figure:SNR_SMALL1}-B, 
    demonstrating that using minimal differential degree without updates is sub-optimal.
\item $SNR<1$: A \textit{super-linear} error reduction is observed for both MEMC and Random. However, MEMC consistently has a significant advantage (of about 7\%) over Random due to selecting the cascades introducing a larger model change.
\item $SNR<1$: the sparse regime has deeper and longer cascades phase than the logarithmic regime, in line with the sparse SBM having more mass localized at zero differential degree.
\item $SNR<1$: The super-linear phase is followed by a linear phase as in the SNR$>1$ regime. At 50\% queries MEMC's error curve aligns with the type-2 error curve. Perfect recovery is achieved at 80\% for MEMC vs. 100\% for Random.
\end{enumerate}
\noindent \textbf{Performance experiments.}
We present experiments for 2- and multi-community detection on SBM networks and real-world social networks (see data links and details at \cite{stanford_socdata}). The authors suggest that the student’s class years form interesting ground truth communities in Facebook social network. We test the ability of algorithms ‘Active,’ and various active learners to cluster Colgate University students according to the 4 class years 2006-09 and Harvard University students according to the 6 class years 2004-09.

The performance of MEMC (`MEMC-Active') is compared with the two known active learners developed for SBMs \cite{Moore} (`MI') and \cite{Gadde} ('MinDeg'), and with the graph-based active learning algorithm of \cite{Kushnir} (`LAAL'). A random selection ('Random') and an Entropy ('Entropy-Active) criterion using our SDP are also presented as baselines to MEMC's criterion.
The advantage of MEMC is prominent: 'MI' uses max mutual information of a query with the current labeling and therefore has sub-optimal performance when labelling is still scarce. MinDeg uses the minimal differential degree without updates (applicable only for SNR$>1$). 'Entropy' uses our SDP formulation, however, the max-entropy criterion focuses on uncertainty and does not explore nodes that can have more extensive model change. Thus only towards later stages of refinement the entropy provides improved gains. Finally, 'LAAL' assumes smoothness of the labeling function over the graph which does not exist for SBMs.


\newpage
\onecolumn

\appendix
\section{Proofs}
\subsection{Theorem \ref{conj:detect_thresh}}
\begin{proof}
see \cite{Massoulie,Mossel_1} for the details of the proof.
\end{proof}


\subsection{Theorem \ref{Theorem:Main}}
\begin{proof}

We start with some definitions:

\begin{definition}\label{yin}
For a given adjacency matrix $\textbf{M} = M$, let $e=e(M)$ be the set of edges in $M$. Also, let $|e|$ be the size of the set $e$ (i.e. the total number of edges).
\end{definition}

\begin{definition}
Given the complete labeling assignment $X\in\bigtriangleup_r^n$, let $e_{in}(X)$ be the number of edges where both endpoint nodes have the same label according to $X$. Then, since each $X_i$ is a unit vector, \begin{equation}\label{in}
e_{in}(X)=\frac{1}{2}\sum_{\substack{X_i=X_j\\(i,j)\in e}}\langle X_i, X_j\rangle.
\end{equation}
\end{definition}

\begin{definition}
Given the complete labeling assignment $X\in\bigtriangleup_r^n$, let $e_{out}(X)$ be the number of edges where both endpoint nodes have different labels according to $X$. Then, since each $X_i$ lies on the simplex, \begin{equation}\label{out}
e_{out}(X)=-\frac{r-1}{2}\sum_{\substack{X_i\neq X_j\\(i,j)\in e}}\langle X_i, X_j\rangle.
\end{equation}
\end{definition}

\begin{definition}
Given the complete labeling assignment $X\in\bigtriangleup_r^n$, let $g_u(X)$ be the number of nodes assigned to the $u^{th}$ label-vector.
\end{definition}

\begin{remark}\label{AnyLabeling}
It is helpful to notice that $\sum_{u}{g_u(X) \choose 2}$ is the total number of within-group pairs of nodes and $\sum_{u<v}[g_u(X)g_v(X)]$ is the total number of between-group pairs of nodes given the labeling $X$.  From this we see that the following equalities hold for any labeling assignment $X$:
\begin{equation}\begin{split}
&\sum_{u}{g_u(X) \choose 2}-e_{in}(X) = \frac{1}{2}\sum_{\substack{X_i= X_j\\(i,j)\notin e}}\langle X_i, X_j\rangle,\\
&\sum_{u<v}[g_u(X)g_v(X)]-e_{out}(X) = -\frac{r-1}{2}\sum_{\substack{X_i\neq X_j\\(i,j)\notin e}}\langle X_i, X_j\rangle,\\
&{n \choose 2}-\sum_{u<v}[g_u(X)g_v(X)]=\sum_{u}{g_u(X) \choose 2}.
\end{split}
\end{equation}
\end{remark}

From the definition of the SBM we first notice that, unconditioned on a specific adjacency matrix $M$, $\mathbb{P}[\textbf{X}=X]=r^{-n}$ for any $X\in\bigtriangleup^n_{r}$.  However, given a specific SBM-generated adjacency matrix $M$,
\begin{equation*}\begin{split}
\mathbb{P}[\textbf{X}=X|\textbf{M}=M]=&\frac{\mathbb{P}[\textbf{M}=M|\textbf{X}=X]\mathbb{P}[\textbf{X}=X]}{\mathbb{P}[\textbf{M}=M]}=C'\mathbb{P}[\textbf{M}=M|\textbf{X}=X]\\
=&C'\big(p\big)^{e_{in}(X)}\big(1-p\big)^{\sum_{u}{g_u(X) \choose 2}-e_{in}(X)}\\
&\indent\indent\indent\cdot\big(q\big)^{e_{out}(X)}\big(1-q\big)^{\sum_{u<v}[g_u(X)g_v(X)]-e_{out}(X)}\\
=&C'\big(p\big)^{|e|-e_{out}(X)}\cdot \big(1-p\big)^{{n \choose 2}-|e|-\sum_{u<v}[g_u(X)g_v(X)]+e_{out}(X)}\\
&\cdot\big(q\big)^{e_{out}(X)}\big(1-q\big)^{\sum_{u<v}[g_u(X)g_v(X)]-e_{out}(X)}
\end{split}
\end{equation*}
where $C'$ is a constant independent of $X$, the first equality follows from Bayes' Law and the fourth equality follows from remark \ref{AnyLabeling}.

Since, conditioned on $M$, the values for $n$, $e$, $p$ and $q$ are all independent of $X$ we incorporate them into the constant terms $C$ and $C''$ to get,
\begin{equation*}\begin{split}
\mathbb{P}[\textbf{X}=X|&\textbf{M}=M]\\
=C''\bigg[& \Big(\frac{p}{q}\Big)^{-e_{out}(X)}\cdot \Big(\frac{1-p}{1-q}\Big)^{-\big(\sum_{u<v}[g_u(X)g_v(X)]-e_{out}(X)\big)}\bigg]\\
=C''\Bigg[&\exp\bigg(-e_{out}(X)\log\Big(\frac{p}{q}\Big)-\Big(\sum_{u<v}[g_u(X)g_v(X)]-e_{out}(X)\Big)\log\Big(\frac{1-p}{1-q}\Big)\bigg)\Bigg]\\
=C''\Bigg[&\exp\bigg(-r e_{out}(X)\log\Big(\frac{p}{q}\Big)-r\Big(\sum_{u<v}[g_u(X)g_v(X)]-e_{out}(X)\Big)\log\Big(\frac{1-p}{1-q}\Big)\bigg)\Bigg]^{\frac{1}{r}}\\
=C''\Bigg[&\exp\bigg((r-1)(-|e|+e_{in}(X))\log\Big(\frac{p}{q}\Big)-e_{out}(X)\log\Big(\frac{p}{q}\Big)\\
&+(r-1)\Big(-{n \choose 2}+|e|+\sum_{u}{g_u(X) \choose 2}-e_{in}(X)\Big)\log\Big(\frac{1-p}{1-q}\Big)\\
&-\Big(\sum_{u<v}[g_u(X)g_v(X)]-e_{out}(X)\Big)\log\Big(\frac{1-p}{1-q}\Big)\bigg)\Bigg]^{\frac{1}{r}}\\
=C\Bigg[&\exp\bigg(e_{in}(X)\Big((r-1)\log\Big(\frac{p}{q}\Big)\Big)-e_{out}(X)\log\Big(\frac{p}{q}\Big)\\
&+\Big(\sum_{u}{g_u(X) \choose 2}-e_{in}(X)\Big)\big((r-1)\log\Big(\frac{1-p}{1-q}\Big)\big)\\
&-\Big(\sum_{u<v}[g_u(X)g_v(X)]-e_{out}(X)\Big)\log\Big(\frac{1-p}{1-q}\Big)\bigg)\Bigg]^{\frac{1}{r}}.
\end{split}
\end{equation*}
Now, from equations (\ref{in}) and (\ref{out}) and remark \ref{AnyLabeling} we get
\begin{equation*}\begin{split}
\mathbb{P}[\textbf{X}=X|&\textbf{M}=M]\\
=C\Bigg[&\exp\bigg(\Big(\sum_{(i,j)\in e}\langle X_i, X_j\rangle\Big)\Big(\frac{r-1}{2}\log\Big(\frac{p}{q}\Big)\Big)+\Big(\sum_{(i,j)\notin e}\langle X_i, X_j\rangle\Big)\Big(\frac{r-1}{2}\log\Big(\frac{1-p}{1-q}\Big)\Big)\bigg)\Bigg]^{\frac{1}{r}}.\\
=C&\exp\Big(\frac{r-1}{2r}\sum_{(i,j)}M_{i,j}\langle X_i, X_j\rangle\Big)=Ce^{\frac{r-1}{2r}Tr(X^TMX)}
\end{split}
\end{equation*}
as desired. 
\end{proof}

\subsection{Corollary \ref{Corollary:Supervised}}
\begin{proof}
Notice that $\mathbb{P}[\textbf{X}_U=X_U|\textbf{X}_L=X_L]=r^{-(n-k)}$ and $\mathbb{P}[\textbf{M}=M|\textbf{X}_L=X_L]$ are both independent of $X_U$. Thus, by Bayes' theorem,  $\mathbb{P}[\textbf{X}_U=X_U|\textbf{M}=M,\textbf{X}_L=X_L]\propto\mathbb{P}[\textbf{M}=M|\textbf{X}_U=X_U,\textbf{X}_L=X_L]$ and the rest follows from the proof of Theorem 1.
\end{proof}

\subsection{Lemma \ref{lemma:Conditional}}
\begin{proof}
Define $M_{\neg i}$ to be the matrix $M$ with the $i^{th}$ row and column removed.  Then, from Theorem 1, the symmetry of $M$ and the linearity of the trace function we get,
\begin{equation}\begin{split}
\mathbb{P}[\textbf{X}_i=X_i|&\textbf{M}=M,\textbf{X}_L=X_L,\textbf{X}_{U_{\neg i}}=X_{U_{\neg i}}] \\
&= \frac{\exp\Big[{\frac{r-1}{2r}\Big(Tr({X_{U_{\neg i}}\brack X_L}^TM_{\neg i}{X_{U_{\neg i}}\brack X_L})+2Tr(X_i^TM_i{X_{U_{\neg i}}\brack X_L})+Tr(X_i^TM_{(i,i)}X_i)\Big)}\Big]}{\sum_{X_j\in\bigtriangleup_r}\exp\Big[{\frac{r-1}{2r}\Big(Tr({X_{U_{\neg i}}\brack X_L}^TM_{\neg i}{X_{U_{\neg i}}\brack X_L})+2Tr(X_j^TM_i{X_{U_{\neg i}}\brack X_L})+Tr(X_j^TM_{(i,i)}X_j)\Big)}\Big]}\\
& = \frac{\exp{\frac{r-1}{2r}\Big(Tr({X_{U_{\neg i}}\brack X_L}^TM_{\neg i}{X_{U_{\neg i}}\brack X_L})\Big)\exp{\Big(2Tr(X_i^TM_i{X_{U_{\neg i}}\brack X_L})\Big)}\exp{\Big(Tr(X_i^TM_{(i,i)}X_i)\Big)}}}{\exp{\frac{r-1}{2r}\Big(Tr({X_{U_{\neg i}}\brack X_L}^TM_{\neg i}{X_{U_{\neg i}}\brack X_L}\Big)\sum_{X_j\in\bigtriangleup_r}\exp{\Big(2Tr(X_j^TM_i{X_{U_{\neg i}}\brack X_L}\Big)\exp{\Big(Tr(X_j^TM_{(i,i)}X_j}\Big)}}}\\
&=\frac{e^{\frac{r-1}{r}(M_i{X_{U_{\neg i}}\brack X_L}X_i^T)}e^{\frac{r-1}{2r}(M_{(i,i)}X_iX_i^T)}}{\sum_{X_j\in\bigtriangleup_r}e^{\frac{r-1}{r}(M_i{X_{U_{\neg i}}\brack X_L}X_j^T)}e^{\frac{r-1}{2r}(M_{(i,i)}X_jX_j^T)}}\\
&=\frac{e^{\frac{r-1}{r}(M_i{X_{U_{\neg i}}\brack X_L}X_i^T)}}{\sum_{X_j\in\bigtriangleup_r}e^{\frac{r-1}{r}(M_i{X_{U_{\neg i}}\brack X_L}X_j^T)}}.
\end{split}
\end{equation}
where the last equality comes from the fact that $X_jX_j^T=1$ for any $X_j\in\bigtriangleup_r$.  
\end{proof}


\subsection{Theorem \ref{thm:sampComplxHSNR}}
\begin{proof}
We start with providing three necessary auxiliary results to be used in our proof (their proofs are provided in the following subsections \ref{app:auxlem2}, \ref{app:auxcor}, \ref{app:auxlem3}):
First, we prove in Lemma \ref{thm:MEMC} that when the absolute differential degree is similar across all node types (as in $SNR<1$), MEMC has a preference for selecting type-1 nodes over type-2 or correctly assigned nodes.

Next, we prove the conditions under which MEMC will have preference to correct type-2 nodes over querying correctly assigned nodes. In particular, Corollary \ref{corollary:SNRb1_typ2_vs_correct} below guarantees that as long as there are type-2 error nodes with absolute differential degree that is lower than other nodes they will be queried by MEMC:

Lastly, Lemma \ref{lemma:skellam} provides the probability of having minority nodes. While majority nodes can be classified correctly by MAP and ML classifiers once most of their neighbors are correctly identified, minority error nodes (i.e. type-2 error nodes) need to be queried directly in order to be corrected. Therefore estimating their proportion in the overall set is crucial for bounding sample complexity analysis for querying type-2 nodes:

To this end we have the necessary ingredients to provide the sample complexity:  using Lemma \ref{thm:MEMC} we obtain that all type-1 errors are initially corrected by MEMC using $m_1$ queries. At the following stage nodes are queried according to their minimal absolute differential degree following Corollary \ref{corollary:SNRb1_typ2_vs_correct}. During this stage type-2 minority nodes as well as correctly assigned nodes are queried at a frequency depending on their distribution around zero differential degree as function of $a$ and $b$. Based on Lemma \ref{lemma:skellam} we can bound the search space around 0-differential degree for type-2 error nodes with the upper bound in Lemma (\ref{lemma:skellam})
\begin{equation}
\label{eq:upbnd}
P(c_{out}\geq c_{in}) \leq \exp\Big(-\Big(\sqrt{\frac{b}{2}} - \sqrt{\frac{a}{2}}\Big)^2\Big).
\end{equation}
and between $[0,-l_c]$ with sampling the mass equal to the summation of the Skellam probability $P(k;a,b)$. The Skellam distribution models the summation of the two racing Poisson processes with means $a$ and $b$, which forms positive differential degree smaller than $-l_c$:
\begin{equation}
\sum_{k=1}^{-l_c}P(k;a,b),\; where\; P(k;a,b) = e^{-(a+b)}\Big(\frac{a}{b}\Big)^{\frac{k}{2}}I_k(2\sqrt{ab}),
\end{equation}
and $I_k(2\sqrt{ab})$ is the Bessel function of the first order. The degree value $l_c$ can be computed by using the upper bound in Eq. \ref{eq:upbnd}. Specifically, by choosing $a$ and $b$ s.t. the bound is smaller than $o(n^{-1})$ and assigning $l_c = b-a$.
We therefore obtain that the expected sample complexity of MEMC is comprised of sampling first the $m_1$ type-1 errors and then sampling a.a.s all the nodes of absolute differential degree smaller than $-l_c$.

Since the Random criterion selects nodes uniformly at random it will have to sample order $n$ nodes to discover all $m_1$ and $m_2$ nodes.
\end{proof}

\subsection{Lemma \ref{thm:MEMC}}
\label{app:auxlem2}
\begin{proof}
Consider the EMC criterion for some node $v_q$:
\begin{equation}
\begin{split}
\text{EMC}(M,X'_U,X_L,\bigtriangleup_r)_{X_q} =
\sum_{X_q\in \bigtriangleup_r} \hat{\mathbb{P}}[\textbf{X}_q=X_{q}|\textbf{M}=M,\textbf{X}_L=X_L,\textbf{X}_{U_{\neg q}}=X'_{U_{\neg q}}]\cdot
\delta(\Phi,X_q).
\end{split}
\end{equation}
We first focus on the the model change component $\delta(\Phi,X_q) = \|\Phi(M,[X_L,X_q],\tilde{X}) - \Phi(M,X_L,\tilde{X})\|$. We examine the model change for a candidate $q$-node $v_q$ that is a $v^1$-node (type-1 error node) where w.l.g its current label is $\tilde{X}_q=-1$, and its newly assigned label is +1. Assume that $v^1_q$ neighbors  are correctly assigned such that $k+\delta$ are +1 node, and $k$ are -1 nodes \footnote{This assumption can be used by using similar argument to \cite{Mossel_1}: Let $V_{\varepsilon} = {v:d_X(v)<\varepsilon\sqrt{np\log{n}}}$. According to Proposition 4.7 therein no two nodes in $V_{\varepsilon}$ are adjacent}.

The model change $\delta(\Phi,X_q)_{I(v^1_q),I(+1)}$, where $I(\cdot)$ maps the input to its corresponding index in the probability model matrix, will have the following value for changing $v_q$ from its current $-1$ label to $+1$ label (to facilitate notation the denominators in $\Phi$ and the constants in the exponents are omitted):
\begin{equation}
\begin{split}
&\delta(\Phi,\{+1\})_{I(v^1_q),I(+1)}=\big\|\exp\big[\underbrace{\sum_{k+\delta}\log{\frac{p}{q}}}_\text{+1 neighbors} + \underbrace{\sum_{k}\log{\frac{p}{q}}(-1)}_\text{-1 neighbors} + \underbrace{\sum_{n-(k+\delta)}\log{\frac{1-p}{1-q}}}_\text{+1 non-neighbors} + \underbrace{\sum_{n-k}\log{\frac{1-p}{1-q}}(-1)}_\text{-1 non-neighbors}  \big]\\
& \;\;\;\;\;\;\;\;\;\;\;\;\;- \exp\big[\sum_{k+\delta}\log{\frac{p}{q}}(-1) + \sum_{k}\log{\frac{p}{q}} +\sum_{n-(k+\delta)}\log{\frac{1-p}{1-q}}(-1) +\sum_{n-k}\log{\frac{1-p}{1-q}}(-1) \big]\|\\
& \;\;\;\;\;\;\;\;\;\;\;\;\;\;\;\;\;\;\;\;= \big\|\exp\big[\delta \log{\frac{p}{q}} - \delta \log{\frac{1-p}{1-q}}\big] - \exp\big[-\delta \log{\frac{p}{q}} + \delta \log{\frac{1-p}{1-q}}]    \big\|\\
& \Rightarrow \text{EMC}(M,X'_U,X_L,\bigtriangleup_r)_{I(v^1_q),I(+1)} = \exp\big[\delta \log{\frac{p}{q}} - \delta \log{\frac{1-p}{1-q}}\big]\cdot \big\|\exp\big[\delta \log{\frac{p}{q}} - \delta \log{\frac{1-p}{1-q}}\big]\\
 &- \exp\big[-\delta \log{\frac{p}{q}} + \delta \log{\frac{1-p}{1-q}}]    \big\|
\end{split}
\end{equation}
Examining the model change for each of the neighbors $u$ of $v^1_q$, and assuming they do not change their label as a result of the new assignment of $v^1_q$ (and therefore their neighbors do not change their labels either) provides $\text{EMC}(M,X'_U,X_L,\bigtriangleup_r)_{X_u}=0$. Therefore the total model change for a type-1 error node $v^1$ is

\begin{equation}
\begin{split}
\text{EMC}(M,X'_U,[X_L,X^1_q],\bigtriangleup_r) = \exp\big[\delta \log{\frac{p}{q}} - \delta \log{\frac{1-p}{1-q}}\big]\cdot \big\|\exp\big[\delta \log{\frac{p}{q}} - \delta \log{\frac{1-p}{1-q}}\big]\\
- \exp\big[-\delta \log{\frac{p}{q}} + \delta \log{\frac{1-p}{1-q}}]    \big\|
\end{split}
\end{equation}

Next, we examine the model change for a candidate $q$-node $v_q$ that is a $v^2$-node (type-2 error node) where w.l.g its current label is $\tilde{X}_q=-1$, and its newly assigned label is +1. Using similar assumptions on its neighbors we arrive at
\begin{equation}
\begin{split}
\text{EMC}(M,X'_U,[X_L,X^2_q],\bigtriangleup_r)=\exp\big[-\delta \log{\frac{p}{q}} + \delta \log{\frac{1-p}{1-q}}]\cdot \big\|\exp\big[-\delta \log{\frac{p}{q}} + \delta \log{\frac{1-p}{1-q}}\big]\\
 - \exp\big[\delta \log{\frac{p}{q}} -\delta \log{\frac{1-p}{1-q}}]    \big\|
 \end{split}
\end{equation}
To this end we can conclude that, under the above assumptions,
\begin{equation}
\text{EMC}(M,X'_U,[X_L,X^1_q],\bigtriangleup_r) >\text{EMC}(M,X'_U,[X_L,X^2_q],\bigtriangleup_r).
\end{equation}
Next, we attend the model change introduced by flipping the assignment of a node correctly labeled (w.l.g. to +1) to its opposite, resulting in creating a minority node whose model change is similar to that of a $v^2$ node:
\begin{equation}
\begin{split}
\text{EMC}(M,X'_U,[X_L,X^3_q],\bigtriangleup_r)=\exp\big[-\delta \log{\frac{p}{q}} + \delta \log{\frac{1-p}{1-q}}]\cdot \big\|\exp\big[-\delta \log{\frac{p}{q}} + \delta \log{\frac{1-p}{1-q}}\big]\\
 - \exp\big[\delta \log{\frac{p}{q}} -\delta \log{\frac{1-p}{1-q}}]    \big\|
 \end{split}
\end{equation}
We conclude that
\begin{equation}
\text{EMC}(M,X'_U,[X_L,X^3_q],\bigtriangleup_r) =\text{EMC}(M,X'_U,[X_L,X^2_q],\bigtriangleup_r).
\end{equation}
\end{proof}



\subsection{Corollary \ref{corollary:SNRb1_typ2_vs_correct}}
\label{app:auxcor}
\begin{proof}
We use here the result of Lemma \ref{thm:MEMC} where for a given fixed $\delta$ for both nodes
\begin{equation}
\begin{split}
\text{EMC}(M,X'_U,[X_L,X^2_q],\bigtriangleup_r) =\text{EMC}(M,X'_U,[X_L,X^3_q],\bigtriangleup_r) =\\ \exp\big[-\delta \log{\frac{p}{q}} + \delta \log{\frac{1-p}{1-q}}]\cdot \big\|\exp\big[-\delta \log{\frac{p}{q}} + \delta \log{\frac{1-p}{1-q}}\big]
 - \exp\big[\delta \log{\frac{p}{q}} -\delta \log{\frac{1-p}{1-q}}]    \big\|
 \end{split}
\end{equation}
However, for different absolute differential degree such that $\delta_2<\delta_3$ we obtain $\text{EMC}(M,X'_U,[X_L,X^2_q],\bigtriangleup_r)>\text{EMC}(M,X'_U,[X_L,X^3_q],\bigtriangleup_r)$
\end{proof}
\subsection{Lemma \ref{lemma:skellam}}
\label{app:auxlem3}
\begin{proof}
We consider the generation of the edges as a Poisson process, where $c_{out}~Poisson(\frac{b}{2})$ and $c_{in}\sim Poisson(\frac{a}{2})$. Then the difference variable $Z=c_{out}-c_{in}$ Has a Skellam distribution:
$Z \sim Skellam(k;b,a)$ such that
\begin{equation}
P(X=k) = e^{(-(a+b))}(\frac{a}{b})^{\frac{k}{2}}I_k (2\sqrt{ab}),
\end{equation}
where $I_k(z)$ is the Bessel function of first order.
Given that $b<a$ we can use the standard Chernoff bound to prove the upper inequality.
Further noting that $X+Y\sim Poiss(b+a)$ and $X|X+Y\sim Bin(X+Y,\frac{b}{b+a})$, and $P(X>Y)=P(X>\frac{X+Y}{2})$ and upper bounding it by conditioning on $X+Y=i$ we can show that
\begin{equation}
P(X>Y) >\frac{\exp(-(\sqrt{\frac{b}{2}} - \sqrt{\frac{a}{2}})^2)}{(\frac{a+b}{2})^2} - \frac{\exp(-(\frac{b}{2} +\frac{a}{2}))}{\sqrt{2}ab} - \frac{\exp(-(\frac{b}{2} +\frac{a}{2}))}{2ab}
\end{equation}
see more details at \cite{skellam}
\end{proof}


\subsection{Lemma \ref{lemma:cascades}}

\begin{proof}
We first consider the case of a majority node $v_j$ with neighbor $v_i$ which has changed its label from $\tilde{x}_i$ to $\tilde{x}_i^{up}$. We observe the following probabilities
\begin{itemize}
  \item $P\{\tilde{x}_j\neq \tilde{x}_i^{up}\}=\frac{1}{2}$ (having different label than the newly revealed neighbor's label),
  \item $P\{\tilde{x}_j\neq {x}_j\}=\frac{1}{2}$(having an erroneous assignment) at $SNR<1$, and
  \item $P\{x_j=x_i\}=\frac{a-b}{(a+b)}$ (having similar ground truth label as its neighbors flipped label).
\end{itemize}  Therefore, the probability of $v_j$ flipping its current (erroneous) label to its correct label is $\frac{a-b}{4(a+b)}$. In the same pattern we summarize the different label-flip probabilities for majority and minority nodes, given a label-flip at a neighboring node:

\begin{center}
\begin{tabular}{ |c|c|c| }
 \hline
 node type & correct flip & incorrect flip \\\hline\hline
 majority & $\frac{a}{4(a+b)}$ & $\frac{b}{4(a+b)}$ \\\hline
 minority & $\frac{b}{4(a+b)}$ & $\frac{a}{4(a+b)}$ \\
 \hline
\end{tabular}
\end{center}
To this end we can compute the expected number of nodes to correctly change their label following a query
\begin{equation}\label{eq:corr_filps}
  d\cdot\Big( \bar{p}_{maj}\frac{a}{4(a+b)} - \bar{p}_{maj}\frac{b}{4(a+b)} \Big) = d \cdot\bar{p}_{maj}\Big(\frac{a-b}{4(a+b)}\Big) = d \cdot p_{maj},
\end{equation}
where $\bar{p}_{maj} = 1-2\bar{p}_{min}$, $\bar{p}_{min}$ is defined as the upper bound in Lemma \ref{lemma:skellam}:
\begin{equation}
P(c_{out}\geq c_{in}) \leq \exp\Big(-\Big(\sqrt{\frac{b}{2}} - \sqrt{\frac{a}{2}}\Big)^2\Big),
\end{equation}
and $p_{maj} = \bar{p}_{maj}\Big(\frac{a-b}{4(a+b)}\Big).$\\
Next, given the average degree $d$ we consider the cascade of diameter that is $O(log_d(n))$ as the following power series:
\begin{equation}\label{eq:cacade_gseries}
  N_{maj} = dp_{maj} + dp_{maj}dp_{maj}+...+(dp_{maj})^{log_dn} =\frac{dp_{maj}}{1-dp_{maj}}(1-(dp_{maj})^{\log_d(n)}).
\end{equation}
Similar derivation is applied to minority nodes to obtain $N_{min}$.

\end{proof}

\subsection{Theorem \ref{thm:sampComplxLSNR}}

\begin{proof}
The active learning process of MEMC is comprised of 3 stages:\\
1. \textbf{Super-linear cascades phase.} The super-linear phase in which cascades take place poses the highest EMC. This stage concludes once there exists no path of size larger than 2 in which nodes of zero differential degree exist, with respect to the assignment $\tilde{X}$. The expected number of queries to attain this state is obtained by dividing $\frac{n}{2}$ by the number of nodes that have flipped their assignments per query, and as such attained non-zero differential degree. The number of such nodes is derived from Lemma \ref{lemma:cascades} as $N_{maj}+N_{min}$. Therefore, the expected number of queried nodes in this stage is the first component in  Eq. (\ref{eq:sampComplxLSNR_MEMC}):
\begin{equation}\label{}
  \frac{n}{2(N_{maj}+N_{min})}.
\end{equation}
2. \textbf{Local type-1 node queries.} The local model change of type-1 error nodes suggest the next highest EMC. As suggested in Lemma \ref{thm:MEMC} and observed for the $SNR>1$ case. The local type-1 node queries starts once there exists no path of size larger than 2 in which nodes of zero differential degree exist (which typically gives rise to cascades). We therefore subtract from the existing $m_1$ errors the type-1 error nodes that have been corrected via the cascades process, and since MEMC will query only type-1 we consider this difference as the set of queries for this stage
\begin{equation}
\Big(m_1 - \frac{nN_{maj}}{2(N_{maj}+N_{min})} \Big)
\end{equation}
The local type-1 error correction is terminated once all type-1 nodes are corrected. \\
3. \textbf{Type-2 bounded search}. The final active querying stage includes querying both type-2 nodes and already correct nodes with minimal absolute differential degree within the $[l_c,-l_c]$ differential degree segment. The process is also equivalent to the process for the $SNR>1$ case following Corollary \ref{lemma:SNRb1_typ2_vs_correct} which establishes preference for type-2 nodes with low absolute differential degree. As in Theorem \ref{thm:sampComplxHSNR}, we use the Skellam probability here to represent the mass of nodes with positive differential degree smaller than $-l_c$ and the upper bound in Lemma (\ref{lemma:skellam}) to cover nodes with negative differential degree down to $l_c$. This mass is taken from the remaining nodes after subtracting prior $m_1$ queries and type-1 nodes have been corrected during the super-linear cascades phase:
\begin{equation}
\Big(n-\frac{nN_{min}}{2(N_{maj}+N_{min})} - m_1\Big)\\ \cdot\Big(\sum_{k=1}^{-l_c}P(k;a,b)+\exp\Big(-\Big(\sqrt{\frac{b}{2}}-\sqrt{\frac{a}{2}}\Big)^2\Big)\Big),
\end{equation}

The Random selection algorithm triggers cascades of correction, similarly to MEMC. However, once all paths of zero differential degree with length $l\geq 2$ have been exhausted, the following process entails uniform unbounded sampling on the remaining mass of nodes, scaling with as $n$ queries.
\end{proof}
\section{Best-Fit Simplex}
\label{BestFitAlg}
We present the following algorithm for finding the best-fit simplex for a given set of unit-vectors.

     \begin{figure}[!htpb]\centering
     \small
       \begin{tabular}{rl} \hline
       & \large \textbf{\textit{bestFitSimplex}}$(X,r)$\\
        &  \textbf{Input}: $X$: set of unit vectors, $r$: , $r$: number of vectors in simplex\\
        & \textbf{Output}: $\bigtriangleup_r$:  best-fit simplex\\
        & 1. $V =$ K-Means($X,r$)\\
        & 2. $\bigtriangleup_r =$ bestFitSDP($V$)\\
        \hline
       \end{tabular}
    \caption{Pseudo-code for \textit{bestFitSimplex}.}
    \label{Algorithm:bestFitSimplex}
    \end{figure}

We provide pseudo-code in Figure \ref{Algorithm:bestFitSimplex}.  In this algorithm K-Means is the well-known algorithm and outputs a set of $r$ vectors.  For the algorithm bestFitSDP we define the $(2r\times 2r)$-matrix $A$ where,
\[
\langle A_{i,j}\rangle=
\begin{cases}
1 &\text{if }i= j+r\\
1 &\text{if }i= j-r\\
0 & other wise.
\end{cases}
\]
Then, bestFitSDP finds the best-fit simplex $\bigtriangleup_r$ by factoring the solution $\mathbb{X}={\bigtriangleup_r \brack V}{\bigtriangleup_r \brack V}^T$  of the following SDP
\begin{equation}\label{Algo:bestFitSDP}
\begin{split}
\text{bestFitSDP}(V,r){: } &\max_{\mathbb{X}} \text{ Tr}(A\mathbb{X})\\
&\text{s.t. } \mathbb{X}_{ii} = 1, \text{ for } 1\leq i\leq 2r\\
&\;\;\;\;\; \mathbb{X}_{ij} = -\frac{1}{r-1} \text{ for } 1\leq i,j\leq r\\
&\;\;\;\;\; \mathbb{X}_{ij} = \langle V_i,V_j\rangle \text{ for } r+1\leq i,j\leq2r\\
&\;\;\;\;\; \mathbb{X} \succeq 0.\\
\end{split}
\end{equation}
We define the output of bestFitSDP$(V,r)$ to be $\bigtriangleup_r$ rotated so that the vectors $V$ in our output ${\bigtriangleup_r \brack V}$ line up with the original input vectors $V$.  This completes the algorithm.

\section{Increased SNR error behaviour}

\begin{figure}[!htpb]
\begin{center}
\includegraphics[height=2.9in, width =3.3in]{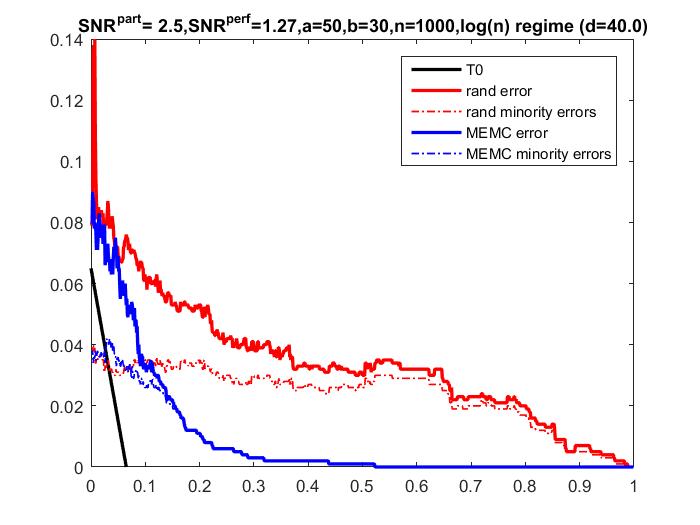}
\end{center}
\caption{High SNR comparison of MEMC error with Random error and the optimal active learner error}
\label{Figure:SNR_HIGH}
\end{figure}

\end{document}